%% file: id_semi-blind_cell-free_arXiv.tex
\documentclass[conference]{IEEEtran}

\usepackage{amsmath,amsfonts,mathtools,amsthm,amssymb}
\usepackage[noend]{algorithmic}
\usepackage{algorithm}
\usepackage{array}
\usepackage[caption=false,font=footnotesize]{subfig}
\usepackage{textcomp}
\usepackage{stfloats}
\usepackage{url}
\usepackage{color}
\usepackage{verbatim}
\usepackage{graphicx}
\usepackage{cite}
\usepackage[hidelinks]{hyperref}
\usepackage[nolist]{acronym}
\usepackage{bbm}
\usepackage[caption=false,font=footnotesize]{subfig}
\def\BibTeX{{\rm B\kern-.05em{\sc i\kern-.025em b}\kern-.08em
    T\kern-.1667em\lower.7ex\hbox{E}\kern-.125emX}}

\makeatletter
\newcommand\fs@betterruled{%
  \def\@fs@cfont{\bfseries}\let\@fs@capt\floatc@ruled
  \def\@fs@pre{\vspace*{5pt}\hrule height.8pt depth0pt \kern2pt}%
  \def\@fs@post{\kern2pt\hrule\relax}%
  \def\@fs@mid{\kern2pt\hrule\kern2pt}%
  \let\@fs@iftopcapt\iftrue}
\floatstyle{betterruled}
\restylefloat{algorithm}
\makeatother

\graphicspath{{figures/}}

\newcommand\plotwidth{0.42}
\newcommand\plotwidthFull{0.32}

\input{abbreviations}
\input{macros}

\newtheorem{proposition}{Proposition}

\newtheorem{lemma}{Lemma}
\newtheorem{remark}{Remark}

\begin{document}
\title{On the Identifiability of Semi-Blind Estimation in Cell-Free Massive MIMO Networks
}

\author{\IEEEauthorblockN{Christian Forsch and Laura Cottatellucci}
\IEEEauthorblockA{Institute for Digital Communications, Friedrich-Alexander-Universität Erlangen-Nürnberg, Erlangen, Germany \\
Email: \{christian.forsch, laura.cottatellucci\}@fau.de}
}

\maketitle

\begin{abstract}
Semi-blind \ac{JCD} is a promising approach to mitigate pilot contamination in \ac{CF-MaMIMO} networks.
The effectiveness of such methods fundamentally depends on identifiability, i.e., the ability to unambiguously recover the unknown channel coefficients and transmitted data signals from the received uplink observations.
In this work, we investigate the identifiability of semi-blind \ac{JCD} from a large-scale system design perspective.
We consider a \ac{CF-MaMIMO} network in which \acp{AP} and \acp{UE} are spatially distributed according to \acl{PPP}es (\acs{PPP}s\acused{PPP}).
The resulting network topology is modeled as \ac{BRGG} that captures local connectivity induced by wireless propagation.
To enable a tractable analysis, the spatially dependent graph model is approximated by a surrogate independent-edge random graph with matched degree distributions.
Building on this model, we develop a recursive probabilistic analysis that characterizes the conditions under which semi-blind recovery succeeds with high probability.
The proposed analysis reveals an identifiability region as a function of key system parameters, including \ac{UE} and \ac{AP} densities and the connectivity radius beyond which channel coefficients are assumed negligible.
Monte Carlo simulations validate the predicted identifiability region and assess the accuracy of the proposed graph approximation.
The proposed framework provides system level insights into how network density and connectivity affect identifiability in large-scale \ac{CF-MaMIMO} systems and offers guidelines for selecting deployment parameters and pilot sequence lengths that enable reliable semi-blind recovery.
\end{abstract}


\acresetall

\vspace*{-1mm}
\section{Introduction}\label{sec:intro}
\Acl{CF-MaMIMO} (CF-MaMIMO)\acused{CF-MaMIMO}\acused{MaMIMO}\acused{MIMO} has emerged as a promising network architecture for next-generation wireless communication systems due to its ability to provide uniformly high service quality and energy-efficient communications through the distributed deployment of a large number of \acp{AP} that jointly serve \acp{UE}~\cite{Ngo2017,Ngo2018,Mohammadi2024}.
Accurate \ac{CSI} is essential to realize the potential gains of \ac{CF-MaMIMO} systems.
However, conventional pilot-based channel estimation schemes suffer from pilot contamination due to the large number of \acp{UE} and the limited availability of pilot signaling resources in practice.
To address this challenge, semi-blind estimation techniques, which jointly exploit limited pilot information and data signals, have gained increasing attention~\cite{Gholami2021a,Gholami2021b,Karataev2024,Zhao2024,Forsch2025,Zhong2024,Yang2025}.
These methods enable reduced pilot overhead while improving robustness against pilot contamination.

Despite their potential, the effectiveness of semi-blind estimation methods critically depends on the identifiability of the underlying bilinear system~\cite{Carvalho2004}.
In this context, identifiability refers to the ability to uniquely recover channel parameters and transmitted data symbols from the noise-free received signal.
In~\cite{Gholami2021a}, sufficient and necessary identifiability conditions in \ac{CF-MaMIMO} networks were derived within the framework of deterministic identifiability.
The network was represented by a bipartite graph whose two sets of nodes correspond to \acp{UE} and \acp{AP}, respectively.
It was shown that a sufficient identifiability condition is satisfied if the first phase of the Karp-Sipser procedure~\cite{Karp1981,Aronson1998} yields an empty \emph{core} in the associated bipartite graph.%
\footnote{
Here, the Karp-Sipser core is the residual subgraph obtained after iterative removal of \emph{\ac{AP}} leaf nodes and their neighboring \ac{UE} nodes.
Note that originally the Karp-Sipser algorithm iteratively removes \emph{any} leaf node and its neighbor.
}

The size of the Karp-Sipser core has been intensively studied within the framework of the \emph{maximum matching} problem on random graphs.
Existing works~\cite{Karp1981,Aronson1998,Bauer2001,Glasgow2025,Frieze2012} have mainly focused on random graph models with independently generated edges.
In~\cite{Karp1981,Aronson1998,Bauer2001,Glasgow2025}, the Karp-Sipser core size for a fundamental graph model with independently generated edges, called Erd\H{o}s–R\'{e}nyi graphs, has been studied.
The authors in~\cite{Frieze2012} analyzed the maximum matching problem in bipartite graphs with a fixed number of independently generated edges per vertex and derived conditions under which the first phase of the Karp-Sipser algorithm finds a maximum matching, i.e., the Karp-Sipser core is empty.
However, \ac{CF-MaMIMO} networks are inherently governed by spatial geometry and are more accurately described by geometric graphs~\cite{Penrose2003}.
In particular, \ac{CF-MaMIMO} networks are modeled by \acp{BRGG}, also known as AB random geometric graphs~\cite{Penrose2014,Dereudre2018}.
In contrast to independent-edge graphs, for \acp{BRGG}, a theory describing the Karp-Sipser core size and its asymptotic behavior remains largely unexplored.

In this work, we analyze the identifiability properties of semi-blind estimation in \ac{CF-MaMIMO} under spatially distributed \acp{AP} and \acp{UE}, modeled by \acl{PPP}es (\acs{PPP}s\acused{PPP}).
The resulting network is represented by a \ac{BRGG} whose spatial dependencies make the asymptotic analysis of identifiability challenging.
To address this issue, we approximate the \ac{BRGG} by a surrogate independent-edge random graph with matched degree distributions, which enables a recursive probabilistic analysis that tracks the evolution of the probability that channels and transmitted data become identifiable during the successive node removal via the Karp-Sipser procedure.
Based on this framework, we characterize the identifiability region as a function of key network parameters, i.e., \ac{UE} and \ac{AP} densities, connectivity radius, and pilot sequence length.
Monte Carlo simulations are provided to validate the predicted identifiability region and assess the accuracy of the proposed surrogate graph model. The developed framework offers insights into the interplay between network geometry and identifiability, providing useful guidelines for the design and dimensioning of large-scale CF-MaMIMO systems.

\vspace*{-1mm}
\section{System Model}\label{sec:sys_mod}
We consider the uplink of a \ac{CF-MaMIMO} network consisting of $L$ geographically distributed single-antenna \acp{AP} and $K$ synchronized single-antenna \acp{UE} with $L\geq K$ in a $D\times D$ square area.
The $L$ \acp{AP} are connected to a \ac{CPU} via fronthaul links.
The channel between \ac{AP} $l$ and \ac{UE} $k$ is described by the channel coefficient $h_{lk}$.
Due to the path loss and the distributed nature of \ac{CF-MaMIMO} networks, many of the elements in the channel matrix $\lmat{H}\!\in\!\Cset^{L\times K}$, with $(l,k)$-element $h_{lk}$, are negligible.
Based on the large-scale fading of the channel between \ac{AP} $l$ and \ac{UE} $k$, we assume the channel coefficient $h_{lk}$ to be negligible if \ac{AP} $l$ and \ac{UE} $k$ are separated by a distance higher than a given threshold $\gamma$.
We decompose the channel matrix $\lmat{H}$ accordingly into two matrices $\lmat{H}_I$ and $\lmat{H}_0$ such that $\lmat{H} = \lmat{H}_I + \lmat{H}_0$, where $\lmat{H}_I$ and $\lmat{H}_0$ denote the matrices of the relevant and negligible channel coefficients, respectively.
Thus, for each \ac{AP} $l$ it is required to estimate only the channels of \acp{UE} that are located in a disc centered around \ac{AP} $l$ with radius $\gamma$ while the signals transmitted from \acp{UE} outside the disc are treated as additive noise.
Throughout this paper, we assume that $\gamma\ll D$ and a uniform distribution of the \acp{AP} and \acp{UE} over the whole network area such that $\lmat{H}_I$ has a large number of zero elements, i.e., $\lmat{H}_I$ is sparse.

During the channel coherence time of $T$ channel uses, the channels are assumed to be constant and each \ac{UE} transmits $T_p$ pilot and $T_d=T-T_p$ data symbols.
The pilot sequences are assumed to be orthogonal and known by the \ac{CPU}.
Hence, there are $T_p$ pilot sequences $\lvec{x}^{(p)}\!\in\!\Cset^{T_p\times1}$, $p=1,\dots,T_p$, and \acp{UE} transmitting the same sequence $\lvec{x}^{(p)}$ are grouped in the set $\mathcal{G}_p$.
The matrix $\lmat{Y}\!\in\!\Cset^{L\times T}$ collects the signals received across all \acp{AP} over $T$ time instants and can be expressed as
\begin{equation}
    \lmat{Y} = \lmat{H}_I\lmat{X} + \lmat{H}_0\lmat{X} + \lmat{N},
    \label{eq:Y}
\end{equation}
where $\lmat{X}=[\pilot{\lmat{X}}\;\data{\lmat{X}}]\!\in\!\Cset^{K\times T}$ is the transmit symbol matrix, consisting of a pilot component $\pilot{\lmat{X}}\!\in\!\Cset^{K\times T_p}$ and a data component $\data{\lmat{X}}\!\in\!\Cset^{K\times T_d}$, and $\lmat{N}\!\in\!\Cset^{L\times T}$ is the matrix of \ac{AWGN}.

\vspace*{-1mm}
\section{Identifiability}\label{subsec:problem}
To achieve a reliable communication performance, especially in the presence of pilot contamination, the \ac{CPU} jointly estimates the channel matrix $\lmat{H}$ and the user data matrix $\data{\lmat{X}}$ given the received signal $\lmat{Y}$ and the pilot matrix $\pilot{\lmat{X}}$.
In this work, we focus on the identifiability of channels and data signals.
We operate within the framework of deterministic identifiability in which the parameters to be estimated are assumed to be deterministic and the channel component $\lmat{H}_0$ is assumed to be zero.
The parameters $\lmat{H}_I$ and $\data{\lmat{X}}$ are said to be identifiable if they can be unambigiously recovered from the noise-free received signal~\cite{Gholami2021a},
\begin{equation}
    \lmat{H}_I\lmat{X} = \tilde{\lmat{H}}_I\tilde{\lmat{X}} \quad\Rightarrow\quad \lmat{H}_I = \tilde{\lmat{H}}_I\quad\text{and}\quad\data{\lmat{X}} = \data{\tilde{\lmat{X}}}.
    \label{eq:ID}
\end{equation}

Sufficient and necessary conditions for identifiability were derived in~\cite{Gholami2021a}.
Furthermore, an iterative  algorithm based on a bipartite graph representation of the \ac{CF-MaMIMO} network was proposed to determine the unknown parameters $\lmat{H}_I$ and $\data{\lmat{X}}$.
It was found that the channel and data matrices are identifiable if all the nodes in the graph are removed through the iterative elimination of leaf/degree-one \ac{AP} nodes and their neighboring \ac{UE} nodes.

\vspace*{-1mm}
\section{Poisson Point Process Networks}\label{sec:PPP_networks}
In this section, we introduce a \ac{CF-MaMIMO} network model in which the spatial deployment of \acp{AP} and \acp{UE} is described by independent \acp{PPP}, hereafter referred to as a \ac{PPP} network.
This framework enables the statistical characterization of the \ac{BRGG} associated with the \ac{CF-MaMIMO} network, allowing the analysis of the iterative leaf node removal that characterizes the identifiability of channel coefficients and data.

We consider a network deployed over a square region of side length $D$ and analyze the random ensemble of \ac{PPP} networks in the asymptotic regime as $D\rightarrow\infty$.
The \acp{AP} and the \acp{UE} are modeled as points of independent homogeneous \acp{PPP}~\cite{Daley2003}.
In particular, the \acp{AP} follow a \ac{PPP} with density $\lambda_R$, corresponding to an average of $\lambda_R$ \acp{AP} per unit area.
\acp{UE} in $\mathcal{G}_p$ are modeled as a homogeneous \ac{PPP} with density $\lambda_T^{(p)}.$
By the superposition property of \acp{PPP}, the union of the \ac{UE} point processes associated with different pilot sequences is also a \ac{PPP} with total density $\lambda_T=\sum_{p=1}^{T_p} \lambda_T^{(p)}$.
We define the $\gamma$-neighborhood of a point $P$ in the \ac{AP} \ac{PPP} as the set of all points in the \ac{UE} \ac{PPP} located within a distance not greater than $\gamma$ from $P$.
Similarly, we define the $\gamma$-neighborhood of a point $P$ in the \ac{UE} \ac{PPP}.

The \ac{UE} and \ac{AP} \acp{PPP}, together with the notion of a $\gamma$-neighborhood, induce a random bipartite graph in which the two disjoint node sets correspond to the \ac{AP} \ac{PPP} and the \ac{UE} \ac{PPP}.
An edge exists between an \ac{AP} node and a \ac{UE} node if and only if their distance is at most $\gamma$ or, equivalently, the \ac{AP} lies within the $\gamma$-neighborhood of the \ac{UE} and the \ac{UE} lies within the $\gamma$-neighborhood of the \ac{AP}.
The resulting ensemble of random bipartite graphs is referred to as the standard ensemble of \ac{PPP} networks with parameters $(\lambda_T^{(p)}, \lambda_R, \gamma)$.
Based on standard properties of \acp{PPP}, we can state the following results:
\begin{itemize}
    \item For any bounded set $\mathcal{A}\subset\Rset^2$, the number of points $N(\mathcal{A})$ in $\mathcal{A}$ from a \ac{PPP} with density $\lambda$ follows a Poisson distribution,
    \vspace*{-2mm}
    \begin{equation}
        \mathrm{Pr}(N(\mathcal{A})=k) = \frac{\er^{-\Lambda} \Lambda^k}{k!},
        \label{eq:N_A}
        \vspace*{-1mm}
    \end{equation}
    with parameter $\Lambda=\lambda |\mathcal{A}|$ where $|\mathcal{A}|$ is the area of $\mathcal{A}.$
    \item From the previous result, the number of \acp{AP} in the  $\gamma$-neighborhood of a given \ac{UE} is Poisson distributed with parameter $\Lambda_R=\lambda_R\pi\gamma^2$.
    This random variable coincides with the edge degree of the corresponding \ac{UE} node.
    \item Similarly, the number of \acp{UE} in $\mathcal{G}_p$ in the $\gamma$-neighborhood of a given \ac{AP} follows the Poisson distribution with parameter $\Lambda_T^{(p)}=\lambda_T^{(p)}\pi \gamma^2$.
    This random variable coincides with the edge degree of the associated \ac{AP} node.
    \item As the average number of edges originated in a \ac{UE} node is $\lambda_R\pi\gamma^2$ and the average number of \acp{UE} in $\mathcal{G}_p$ in a \ac{PPP} network of area $D^2$ is $\lambda_T^{(p)}D^2$, the average number of edges in the corresponding bipartite graph is $\pi\gamma^2\lambda_R\lambda_T^{(p)}D^2$.
    The average number of edges scales linearly with the average number of \ac{UE} nodes $\lambda_T^{(p)} D^2$ or, alternatively, with the average number of \ac{AP} nodes $\lambda_R D^2$, implying that the bipartite random graph is \emph{sparse}. 
\end{itemize}
The bipartite random graphs generated as described above are characterized by independently generated nodes, while the edges exhibit intrinsic correlations induced by the underlying geometric constraints.
They are known in the literature as \acp{BRGG} \cite{Penrose2003} or AB random geometric graphs \cite{Penrose2014}.

\vspace*{-1mm}
\section{Asymptotic Identifiability for PPP Networks}\label{sec:ID_PPP}
In this section, we investigate the identifiability of the standard ensemble of \ac{PPP} networks defined by the parameters $(\lambda_T^{p},\lambda_R,\gamma)$ in the asymptotic regime where the network dimensions tend to infinity, i.e., $D\rightarrow\infty$.
The objective is to determine whether the bilinear inverse problem associated with the joint estimation of channel coefficients and transmitted data is asymptotically identifiable with high probability, independently of a specific network realization.
In particular, identifiability is determined solely as a function of macroscopic network parameters, such as the \ac{UE} and \ac{AP} spatial intensities and the neighborhood radius.

To this end, we apply the  algorithm for identifiability proposed in~\cite{Gholami2021a}.
It decomposes the identification assessment into $T_p$ identification procedures.
Each procedure focuses on the \acp{UE} in the set $\mathcal{G}_p$ with $p=1,\dots,T_p$ and determines whether the corresponding unknown parameters in $\lmat{H}_I$ and $\data{\lmat{X}}$, are identifiable.
Such parameters are  collected in the sub-matrices $\lmat{H}_{I,p}$ and $\data{\lmat{X}}_p$. 
The algorithm starts by initializing each \ac{AP} node with the sum of the channel coefficients of all \acp{UE} located within its $\gamma$-neighborhood, which can be inferred from the pilot transmission.
In the first iteration, the algorithm selects all \ac{AP} nodes that are leaves, i.e., connected to a single \ac{UE}.
For each such node, the corresponding channel coefficients and transmitted data symbols can be uniquely recovered.
The knowledge of the transmitted data further enables the recovery of all channel coefficients between the corresponding transmitting \ac{UE} and the \acp{AP} in its $\gamma$-neighborhood; see~\cite{Gholami2021a} for details.
Once identified, these parameters can be removed from the sub-matrices $\lmat{H}_{I,p}$ and $\data{\lmat{X}}_p$ or, equivalently, from the bilinear system of equations while the identified channel coefficients can be subtracted from the aggregated quantities computed during initialization.
From a graph-theoretic perspective, this operation is equivalent to removing the leaf \ac{AP} node, the associated \ac{UE} node, and its incident edges from the bipartite graph.
The resulting residual graph may in turn contain new leaf \ac{AP} nodes, and the procedure is iteratively repeated until no leaf \ac{AP} nodes remain.
The system is identifiable if, for all $p=1,\dots,T_p$, the procedure removes all columns from the submatrices $\lmat{H}_{I,p}$ and $\data{\lmat{X}}_p$, corresponding to an empty residual graph, i.e., the Karp-Sipser core is empty.

Unfortunately, the current theoretical understanding of \acp{BRGG} is not sufficiently developed to enable a rigorous analysis of the associated Karp-Sipser core.
The main difficulties stem from the correlation among edges induced by the underlying geometric constraints.
To overcome this limitation, we relax the edge dependencies and approximate the original model by an analytically tractable independent-edge \ac{PPP} network with matched macroscopic statistics, i.e., matched node degree distributions as functions of the parameters $(\lambda_T^p,\lambda_R,\gamma)$.
Furthermore, we observe that networks containing degree-one \ac{UE} nodes exhibit a nonzero probability of non-identifiability since two or more degree-one \ac{UE} nodes connected to the same \ac{AP} are not identifiable.
Therefore, degree-one \ac{UE} nodes are removed from the graph.
From a practical perspective, this corresponds to enforcing a minimum level of macro-spatial diversity for all \acp{UE}.
Additionally, we remove all the \ac{UE} and \ac{AP} nodes with zero degree since they are intrinsically inactive in the network.
The identifiability condition for the resulting \ac{PPP} network is summarized in the following proposition.

\begin{proposition}[Identifiability Condition for a \ac{PPP} Network] \label{prop:PPP_theorem} 
Consider a \ac{PPP} network with \ac{AP} spatial density $\lambda_R$ and \ac{UE} spatial density $\lambda_T^{(p)}$ for the \acp{UE} in group $\mathcal{G}_p$ transmitting the pilot sequence $\mathbf{x}_p^{(p)}$.
Define
\begin{equation}
    \Lambda_T^{(p)}=\lambda_T^{(p)}\pi\gamma^2
    \label{eq:Lambda_T_p}
\end{equation}
as the average number of \acp{UE} in $\mathcal{G}_p$  within a $\gamma$-neighborhood of a given \ac{AP}, and
\begin{equation}
    \Lambda_R= \lambda_R\pi\gamma^2
    \label{eq:Lambda_R}
\end{equation}
as the average number of \acp{AP} within the $\gamma$-neighborhood of a given \ac{UE}.

Assume that \acp{AP} with an empty $\gamma$-neighborhood and \acp{UE} with less than two \acp{AP} in their $\gamma$-neighborhood are neglected, and edges in the associated bipartite graph are generated independently.
Then, in the asymptotic regime where the network dimension $D\rightarrow\infty$, the unknown parameters $\lmat{H}_I$ and $\data{\lmat{X}}$ are identifiable if for each group of users $\mathcal{G}_p$ the following fixed-point equation
\vspace*{-1mm}
\begin{equation}
    f(\Lambda_T^{(p)}, \Lambda_R, z) =\epsilon_{\Delta}\frac{\er^{-\Lambda_R w(z)}-\er^{-\Lambda_R}}{1-\er^{-\Lambda_R}} = z
    \label{eq:z}
    \vspace*{-1mm}
\end{equation}
with $w(z) = \er^{-\tilde{\Lambda}_T^{(p)} z}$, $\epsilon_{\Delta} = 1 - \frac{\tilde{\Lambda}_T^{(p)}}{\er^{\tilde{\Lambda}_T^{(p)}}-1}$, and $\tilde{\Lambda}_T^{(p)}=\Lambda_T^{(p)}(1-\Lambda_R\er^{-\Lambda_R})$ admits a single solution in $z=0$ for $z\in[0, \epsilon_{\Delta}]$.
\end{proposition}

\begin{proof}
We denote by $z^{(\ell)}$ the average probability that a \ac{UE} node remains in the residual graph after the $\ell$-th iteration of the Karp-Sipser procedure or, equivalently, that the corresponding channel parameters and transmitted data remain unresolved.
Considering the independently generated edges and the sparsity of the graph discussed in Section~\ref{sec:PPP_networks}, we characterize the evolution of $z^{(\ell)}$ across iterations, which we refer to as the density evolution of the procedure.
The analysis follows standard density-evolution arguments for iterative peeling processes on sparse random graphs~\cite{Richardson2008}.
In the following, we first compute the \ac{UE} and \ac{AP} node \emph{degree distributions from a node perspective} $B(x)=\sum_{k=0}^\infty B_k x^k$ and $A(x)=\sum_{k=0}^\infty A_k x^k$, which are polynomials with coefficients $B_k$ ($A_k$) that represent the probability that a randomly chosen \ac{UE} (\ac{AP}) node has degree $k$.
The asymptotic evolution of the iterative procedure depends on the degree distribution observed both from the node perspective and from the edge perspective, as commonly done in the analysis of sparse random graphs.
Hence, we also compute the \ac{UE} and \ac{AP} node \emph{degree distributions from an edge perspective} $\beta(x)=\sum_{k=0}^\infty \beta_k x^{k-1}$ and $\alpha(x)=\sum_{k=0}^\infty \alpha_k x^{k-1}$ with coefficients $\beta_k$ ($\alpha_k$) that represent the probability that a randomly chosen edge is connected to a \ac{UE} (\ac{AP}) node of degree $k$.

According to the construction of the bipartite graph of the \ac{PPP} network, a \ac{UE} node has degree $k\!\in\!\{0,1,2,...\}$ with probability $\check{B}_k=\er^{-\Lambda_R}\frac{\Lambda_R^k}{k!}$, while an \ac{AP} node has degree $k\!\in\!\{0,1,2,...\}$ with probability $\check{A}_k=\er^{-\Lambda_T^{(p)}}\frac{\Lambda_T^{(p)k}}{k!}$.
The removal of degree-zero and degree-one \ac{UE} nodes implies a new normalization of the corresponding distribution, i.e.,  a \ac{UE} node has degree $k\!\in\!\{2,3,4,...\}$ with probability $B_k=\frac{\er^{-\Lambda_R}}{1-\er^{-\Lambda_R}(1+\Lambda_R)} \frac{\Lambda_R^k}{k!}$.
The removal of degree-one \ac{UE} nodes affects the \ac{AP} node degree distribution.
Assuming independent edges, an \ac{AP} node has degree $ k\!\in\!\{0,1,2,...\} $ with probability $  \sum_{i=0}^\infty\check{A}_{k+i} \binom{k+i}{k}\check{B}_1^i(1\!-\!\check{B}_1)^k= \er^{-\tilde{\Lambda}_T^{(p)}} \frac{ \tilde{\Lambda}_T^{(p)k}}{k!}, $ i.e., it follows a Poisson distribution with parameter $\tilde{\Lambda}_T^{(p)}$ defined in Proposition~\ref{prop:PPP_theorem} (a detailed derivation is omitted due to space constraints).
Hence, after renormalization due to the removal of degree-zero \ac{AP} nodes, an \ac{AP} node has degree $k\!\in\!\{1,2,3,...\}$ with probability $A_k=\frac{\er^{-\tilde{\Lambda}_T^{(p)}}}{1-\er^{-\tilde{\Lambda}_T^{(p)}}} \frac{\tilde{\Lambda}_T^{(p)k}}{k!}$.

Based on the \ac{UE} and \ac{AP} node degree distributions from a node perspective $B(x)$ and $A(x)$, the \ac{UE} and \ac{AP} node degree distributions from an edge perspective are given by~\cite{Richardson2008}
\vspace*{-1mm}
\begin{equation}
    \beta(x) \!=\! \frac{B'(x)}{B'(1)} \!=\! \frac{\sum_{k=2}^{\infty} k B_k x^{k-1}}{\sum_{k=2}^{\infty} k B_k} \!=\! \frac{\er^{-\Lambda_R(1-x)}-\er^{-\Lambda_R}}{1-\er^{-\Lambda_R}},
    \label{eq:beta}
\end{equation}
and
\vspace*{-2mm}
\begin{equation}
    \alpha(x) = \frac{A'(x)}{A'(1)} = \frac{\sum_{k=1}^{\infty} k A_k x^{k-1}}{\sum_{k=1}^{\infty} k A_k} = \er^{-\tilde{\Lambda}_T^{(p)}(1-x)},
    \label{eq:alpha}
\end{equation}
respectively.
The density evolution function can be derived along the lines of~\cite[Theorem~3.50]{Richardson2008}.
Thus, the average probability that channels and transmitted data symbols of a \ac{UE} cannot be identified after iteration $\ell$ is given by 
\begin{equation}
    z^{(\ell)} = \epsilon_{\Delta} \beta(1-\alpha(1-z^{(\ell-1)})),
    \label{eq:density_evolution}
\end{equation}
which yields the fixed-point equation~\eqref{eq:z} for $z^{(\ell)}=z^{(\ell-1)}=z$ where $\epsilon_{\Delta}=1-A_1$ is the fraction of \ac{AP} nodes that are not removed after the first iteration of the algorithm and $z^{(1)}=\epsilon_{\Delta}$.
The final formulation of Proposition~\ref{prop:PPP_theorem} establishes the condition for the convergence of $z^{(\ell)}$ to zero for $\ell$ sufficiently large~\cite[Theorem~3.59]{Richardson2008}, which is equivalent to the emptiness of the Karp-Sipser core.
\end{proof} 

\vspace*{-1mm}
In Fig.~\ref{fig:density_evolution}, we plot the density evolution assuming a neighborhood radius $\gamma=70$, a \ac{UE} spatial density $\lambda_T^{(p)}=5\cdot10^{-4}$, and various values of the \ac{AP} spatial density, namely, $\lambda_{R}\in\{0.0036,\,0.0072,\,0.0144\}$.
\begin{figure}[t]
    \centerline{\includegraphics[width=\plotwidth\textwidth]{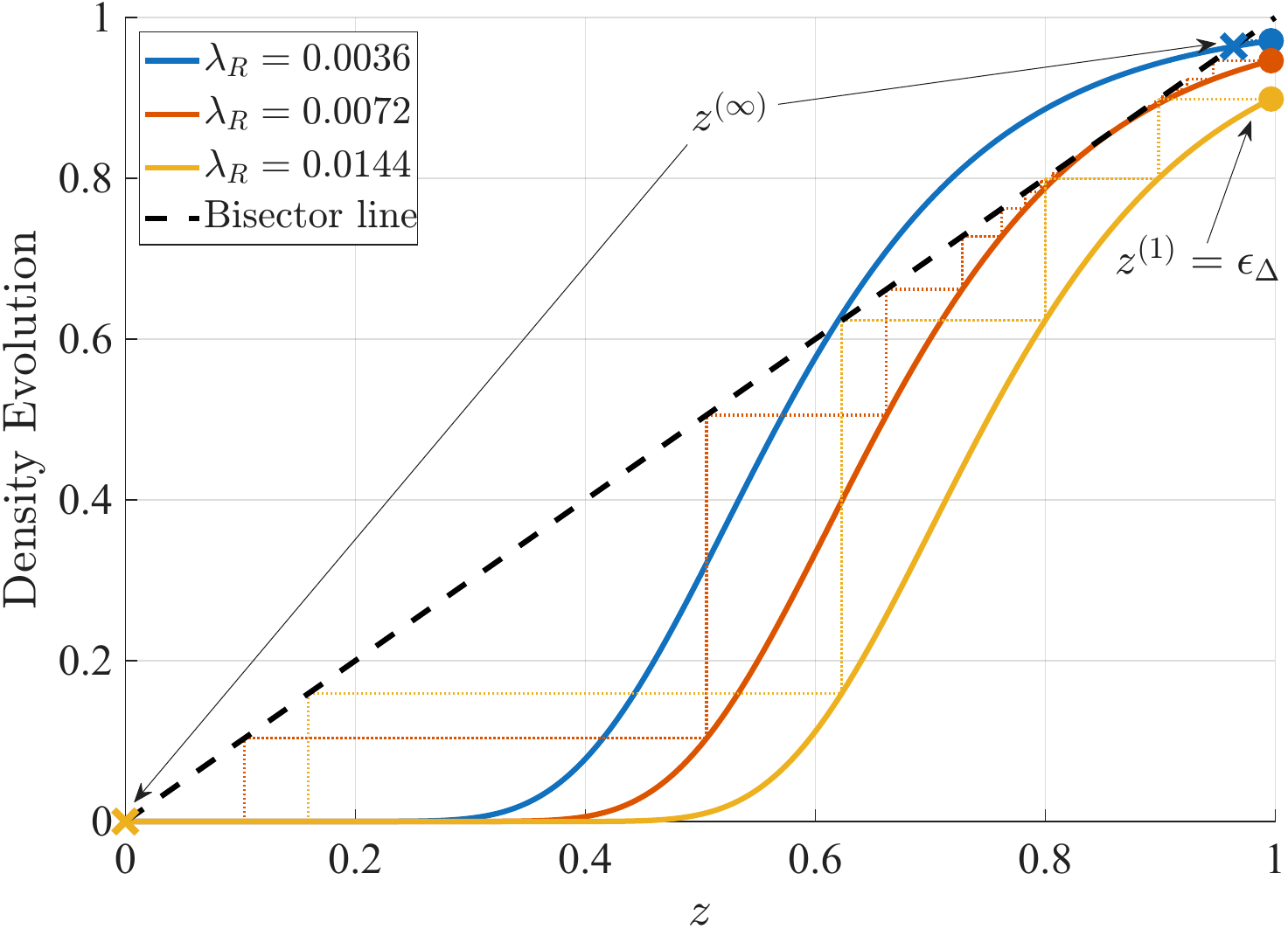}}
    \vspace*{-2mm}
    \caption{Density evolution function for $\lambda_T^{(p)}=5\cdot10^{-4}$ and $\gamma=70$.}
    \vspace*{-4mm}
    \label{fig:density_evolution}
\end{figure}
For the parameters considered in this example, $\tilde{\Lambda}_T^{(p)}\approx7.7$ and $\epsilon_{\Delta}\approx0.9965$ $\forall \lambda_R$, implying that after the initial iteration of the algorithm, represented by a circular marker in Fig.~\ref{fig:density_evolution}, channel coefficients and transmitted data are identifiable for only 0.35\% of the \acp{UE}.
However, after a sufficient number of iterations, shown as a cross marker in Fig.~\ref{fig:density_evolution}, channels and transmitted data of all \acp{UE} can be correctly identified for $\lambda_R=0.0144$, while for $\lambda_R=0.0036$, channels and transmitted data are identifiable only for a very limited fraction of \acp{UE} and the global identifiability condition is not satisfied.
Finally, for an \ac{AP} spatial density of $\lambda_R=0.0072$, a limiting condition appears and a very large number of iterations is required to identify all \ac{UE} parameters in the system.
These observations suggest that, in the asymptotic regime, it is possible to identify a region in the parameter space $(\lambda_T^{(p)}, \lambda_R)$  for which realizations of the \ac{PPP} network are identifiable with high probability, as well as a complementary region where identifiability is not guaranteed.
We refer to the former as the identifiability region for \ac{PPP} networks and investigate its properties in detail in the following section.

\begin{remark}\label{rem:pilot_length}
For a given \ac{AP} \ac{PPP} density $\lambda_R$, Proposition~\ref{prop:PPP_theorem} enables to determine the maximum density $\lambda_T^{(p)}$ for which the \ac{PPP} network remains identifiable.
Thus, it enables to determine the minimum number of pilot sequences $T_p$ requested to ensure the identifiability of channel coefficients and data in a \ac{PPP} network with given intensities $\lambda_R$ and $\lambda_T$.
Alternatively, for a fixed $\lambda_R $ and a maximum number of pilot sequences $T_p$, Proposition~\ref{prop:PPP_theorem} allows to determine the maximum admissible \ac{UE} \ac{PPP} density $\lambda_T$ under the constraint of identifiability.
\end{remark}

\vspace*{-1mm}
\section{Identifiability region of a PPP network}\label{sec:identifiability_region}
In this section, we analyze the fixed-point properties of the density evolution.
The analysis reveals a phase-transition phenomenon: the channel coefficients and transmitted data are identifiable with high probability when the \ac{PPP} network parameters lie within a certain region of the parameter space, whereas identifiability is lost outside this region.
Since the identifiability region based on the density evolution function derived in Section~\ref{sec:ID_PPP} is analytically not tractable, we introduce an approximation that is valid for $\Lambda_R\gg1$, which is well justified for \ac{CF-MaMIMO} systems.
In this approximation, we remove only degree-zero \ac{UE} and \ac{AP} nodes, whereas degree-one \ac{UE} nodes remain in the model.
In this case, the \ac{UE} and \ac{AP} node degree distributions from an edge perspective are given by $\hat{\beta}(x)=\er^{-\Lambda_R(1-x)}$ and $\hat{\alpha}(x)=\er^{-\Lambda_T^{(p)}(1-x)}$, respectively, similar to the \ac{AP} node degree distribution presented in~\eqref{eq:alpha}.
Furthermore, the initial fraction of identifiable channels is $\hat{\epsilon}_\Delta=1-\check{A}_1/(1-\er^{\Lambda_T^{(p)}})$.
The corresponding density evolution function $\hat{f}(\Lambda_T^{(p)}, \Lambda_R, z)=\hat{\epsilon}_{\Delta} \hat{\beta}(1-\hat{\alpha}(1-z))$ is given by
\begin{equation}
    \hat{f}(\Lambda_T^{(p)}, \Lambda_R, z) =  \frac{1-(1+\Lambda_T^{(p)})\er^{-\Lambda_T^{(p)}}}{1-\er^{-\Lambda_T^{(p)}}} \er^{- \Lambda_R \er^{-\Lambda_T^{(p)}z}}.
    \label{eq:density_evolution_function}
\end{equation}
The fixed point of this function characterizes the asymptotic behavior of the Karp-Sipser procedure and, consequently, the identifiability properties of the system. 
Next, we characterize key monotonicity properties of the function $\hat{f}(\Lambda_T^{(p)}, \Lambda_R, z)$, which are instrumental in proving the existence of sharp identifiability thresholds. 
\begin{lemma}
Let $\Lambda_T^{(p)}>0$, $\Lambda_R>0$, and $z\in[0,1]$.
Consider the function $\hat{f}$ in~\eqref{eq:density_evolution_function}.
Then, the following properties hold:
\begin{itemize}
    \item Given $\Lambda_R$ and $z$, $\hat{f}$ is strictly increasing in $\Lambda_T^{(p)}\!\in\!(0,\infty)$.
    \item Given $\Lambda_T^{(p)}$ and $z$, $\hat{f}$ is strictly decreasing in $\Lambda_R\!\in\!(0,\infty)$.
    \item Given $\Lambda_T^{(p)}$ and $\Lambda_R$, $\hat{f}$ is strictly increasing in $z\!\in\![0,1]$.
\end{itemize}
\end{lemma}
\begin{proof}
The results follow directly from computing the partial derivatives of  $\hat{f}(\Lambda_T^{(p)}, \Lambda_R, z)$ with respect to $\Lambda_T^{(p)}$, $\Lambda_R$, and $z$ which are respectively positive, negative, and positive, over the specified domains.
\end{proof}

The monotonicity of $\hat{f}(\Lambda_T^{(p)}, \Lambda_R, z)$ with respect to $\Lambda_T^{(p)}$, for a given value of $\Lambda_R$, guarantees the existence of a \emph{critical value} ${\Lambda_T^{(p)}}^*$, such that the channel parameters and the data are identifiable for $\Lambda_T^{(p)} < {\Lambda_T^{(p)}}^*$, whereas identifiability is lost for $\Lambda_T^{(p)} > {\Lambda_T^{(p)}}^*$.
Similarly, for given $\Lambda_T^{(p)}$, there is a critical value $\Lambda_R^*$ such that identifiability is ensured for $\Lambda_R > \Lambda_R^*$, whereas it is lost for $\Lambda_R < \Lambda_R^*$.
Together, these critical thresholds determine an identifiability phase diagram in the $(\Lambda_T^{(p)},\Lambda_R)$ parameter plane, partitioning the space into identifiable and unidentifiable regions in the large-system limit.

In the following, we fix $\Lambda_T^{(p)}$ and determine the critical threshold $\Lambda_R^*$ that discriminates between identifiable and unidentifiable regimes.
Formally, $\Lambda_R^*$ is defined as
\begin{equation}
    \Lambda_R^* = \inf\{\Lambda_R\in(0,\infty): z^{(\ell)} \rightarrow 0 \;\text{for}\; \ell\rightarrow\infty\},
    \label{eq:Lambda_R_crit_definition}
\end{equation}
where $z^{(\ell)}$ is defined in~\eqref{eq:density_evolution}.
To this end, we characterize the critical point as the value at which the approximate density evolution function becomes tangent to the bisector line.
This condition is captured by the following system of equations,
\begin{equation}
    \hat{f}(\Lambda_T^{(p)}, \Lambda_R, z) = z \quad \text{and} \quad \frac{\partial \hat{f}(\Lambda_T, \Lambda_R, z)}{\partial z} = 1,
    \label{eq:system_of_equations}
\end{equation}
which has the solution,
\begin{align}
    \Lambda_R^* &= -\ln\frac{z^*}{\hat{\epsilon}_\Delta} \cdot \er^{\Lambda_T^{(p)}z^*}
    \label{eq:Lambda_R_crit}, \\
    z^* &= \hat{\epsilon}_\Delta \cdot \er^{W\left(-\frac{1}{\Lambda_T^{(p)}\cdot\hat{\epsilon}_\Delta}\right)},
    \label{eq:z_crit}
\end{align}
where $W(\cdot)$ denotes the principal branch of the Lambert W function.

\begin{remark}
The real-valued Lambert W function $W(x)$ is defined only for $x\geq-\frac{1}{\er}$.
Hence, the solution~\eqref{eq:Lambda_R_crit}-\eqref{eq:z_crit} exists only if $\Lambda_T^{(p)}\cdot\hat{\epsilon}_\Delta\geq\er \;\Leftrightarrow\; \Lambda_T^{(p)}\gtrsim3.1606$.
Otherwise, there are no values of the parameters $\Lambda_R$ and $z$ that yield a function $\hat{f}(\Lambda_T^{(p)}, \Lambda_R, z)$ that is tangent to the bisector line.
However, since the approximation considered in this section is valid for massive \ac{MIMO} systems with $\Lambda_R\gg1$, the condition $\Lambda_T^{(p)}\gtrsim3.1606$ is usually satisfied.
\end{remark}

Rather than expressing identifiability conditions solely in terms of the average number of \acp{UE} or \acp{AP} within a $\gamma$-neighborhood, from a system design perspective, it is more informative to parameterize the model in terms of \ac{UE} and \ac{AP} spatial intensities $\lambda_T^{(p)}$ and $\lambda_R$ along with the neighborhood radius $\gamma$ since they are independent parameters.
Accordingly, the system of equations~\eqref{eq:system_of_equations} can be equivalently expressed in terms of $\lambda_T^{(p)}, \lambda_R$ and $\gamma$ leading to an alternative characterization of the identifiability region. 
To study the impact of the neighborhood radius $\gamma$, we determine the identifiability region in the  $(\lambda_T^{(p)}, \lambda_R)$ plane for different values of the parameter $\gamma$.
Fig.~\ref{fig:id_region} illustrates the resulting  identifiability regions.
\begin{figure}[t]
    \centerline{\includegraphics[width=\plotwidth\textwidth]{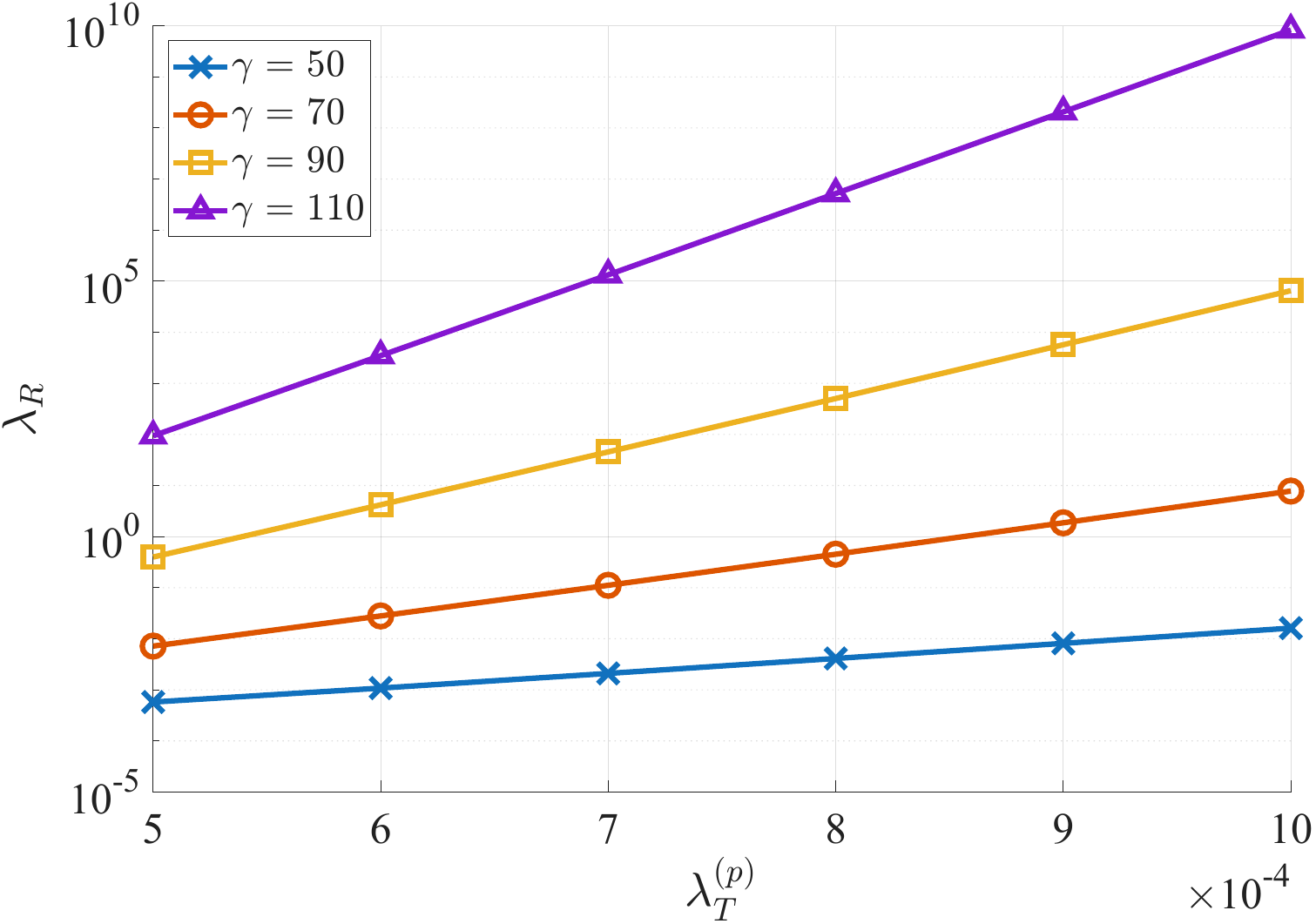}}
    \vspace*{-2mm}
    \caption{Identifiability region $\lambda_{R}$ versus $\lambda_{T}^{(p)}$ for various radii $\gamma$.}
    \vspace*{-4mm}
    \label{fig:id_region}
\end{figure}
All points above the curve correspond to identifiable systems while the points below correspond to unidentifiable configurations.
As shown in Fig.~\ref{fig:id_region}, the identifiability region is convex and becomes larger as the radius $\gamma$ decreases.

\vspace*{-1mm}
\section{Monte Carlo Simulations}\label{sec:sims}
In this section, we present Monte Carlo simulations to validate the results on the identifiability region described in the previous section.
We consider a square network of side length $D=1000\,$m, two values for the neighborhood radius $\gamma\in\{50,\,70\}$, and various \ac{UE} and \ac{AP} intensities $\lambda_T^{(p)}$ and $\lambda_R$.
The values of $\lambda_R$ are set as a function of $\lambda_T^{(p)}$ and $\gamma$ such that $\lambda_R=2^i\cdot\lambda_R^*$ with $i\in\{0,\pm1,\pm2,\pm3\}$, $\lambda_R^*=\Lambda_R^*/(\pi\gamma^2)$, and $\Lambda_R^*$ computed according to~\eqref{eq:Lambda_R_crit}.
Furthermore, we construct the bipartite graph in two different ways: (a) by generating independent edges with matched macroscopic statistics as described in Section~\ref{sec:ID_PPP} and (b) by generating edges between \acp{UE} and \acp{AP} whose pairwise distance is not greater than $\gamma$ as described in Section~\ref{sec:PPP_networks}, yielding a \ac{BRGG}.
The identifiability region derived in Section~\ref{sec:identifiability_region} relies on the assumption of independent edges.
Hence, these simulations allow us to verify the extent to which approximating \acp{BRGG} by independent-edge graphs is justified.
Besides, we neglect all degree-zero nodes as well as degree-one \ac{UE} nodes as described in Section~\ref{sec:ID_PPP}.

For each parameter configuration $(\lambda_T^{(p)}, \lambda_R, \gamma)$, we generate $10^4$ independent network/graph realizations and assess the identifiability of the channels and data signals by applying the Karp-Sipser algorithm.
The results are presented in terms of the \emph{identifiability rate} $r_\text{ID}$ which is defined as the fraction of graphs for which the Karp-Sipser core is empty.
Furthermore, we consider the \emph{per-UE identifiability rate} $r_\text{ID-UE}$ defined as the average fraction of \acp{UE} whose channel parameters and transmitted data are identifiable in a given network.

The results for $\gamma=50\,$m and $\gamma=70\,$m are depicted in Figs.~\ref{fig:r_id_gamma50} and~\ref{fig:r_id_gamma70}, respectively.
\begin{figure*}[t]
    \centering
    \subfloat[$r_\text{ID}$ for independent-edge graphs.]{\includegraphics[width=\plotwidthFull\textwidth]{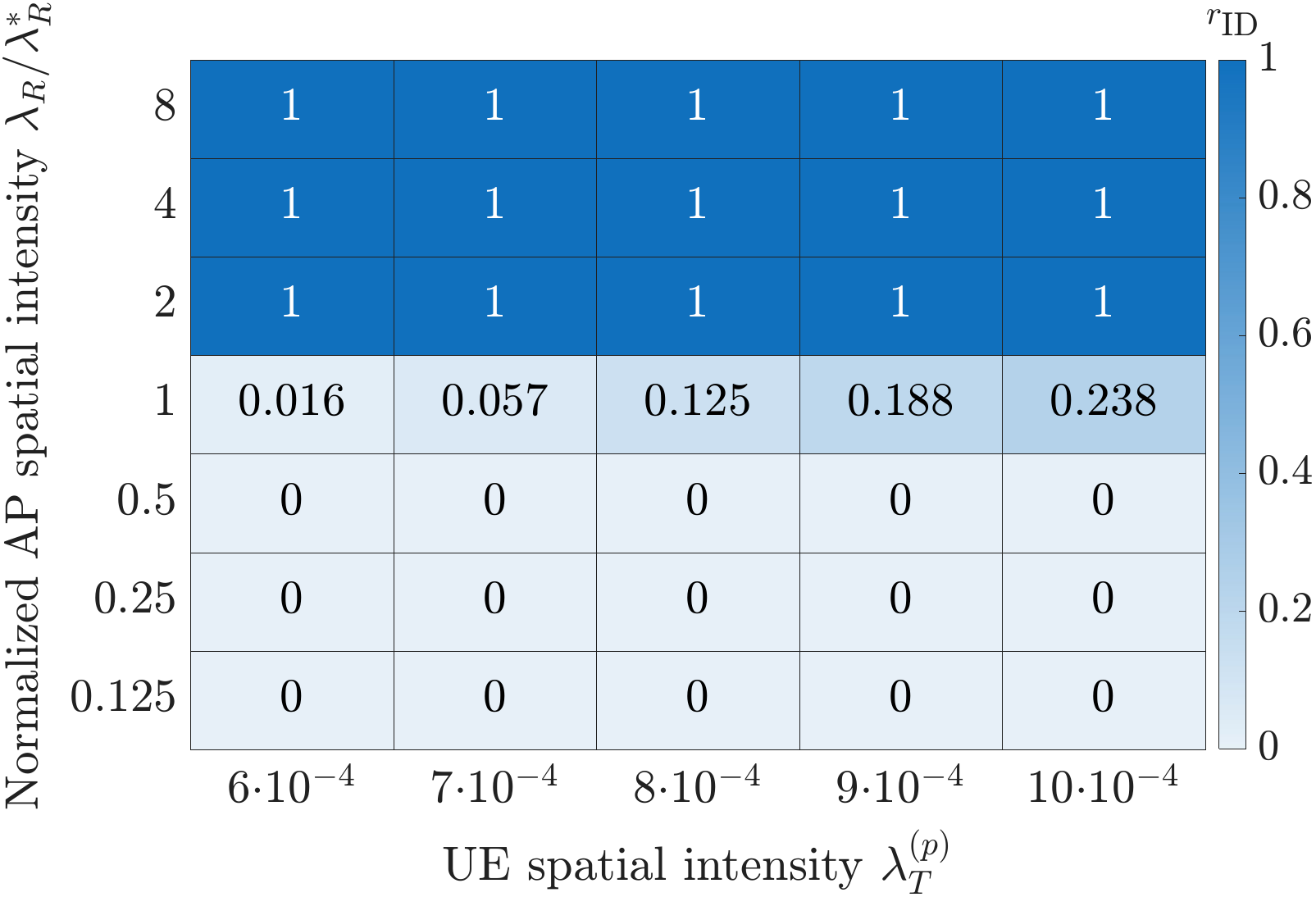}
        \label{fig:r_id_independent_gamma50}}\hfill
    \subfloat[$r_\text{ID}$ for geometric graphs.]{\includegraphics[width=\plotwidthFull\textwidth]{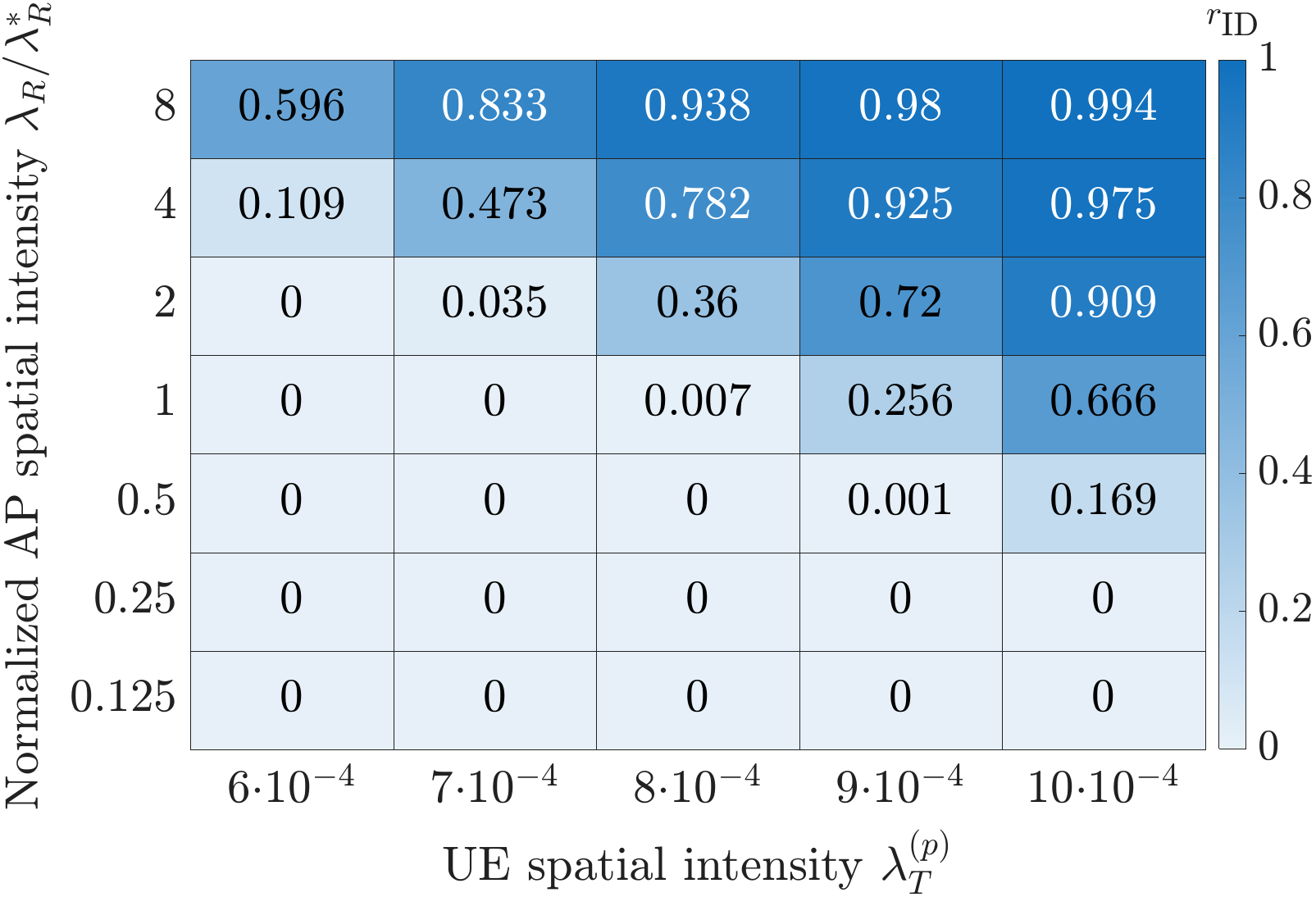}
        \label{fig:r_id_geometric_gamma50}}\hfill
    \subfloat[$r_\text{ID-UE}$ for geometric graphs.]{\includegraphics[width=\plotwidthFull\textwidth]{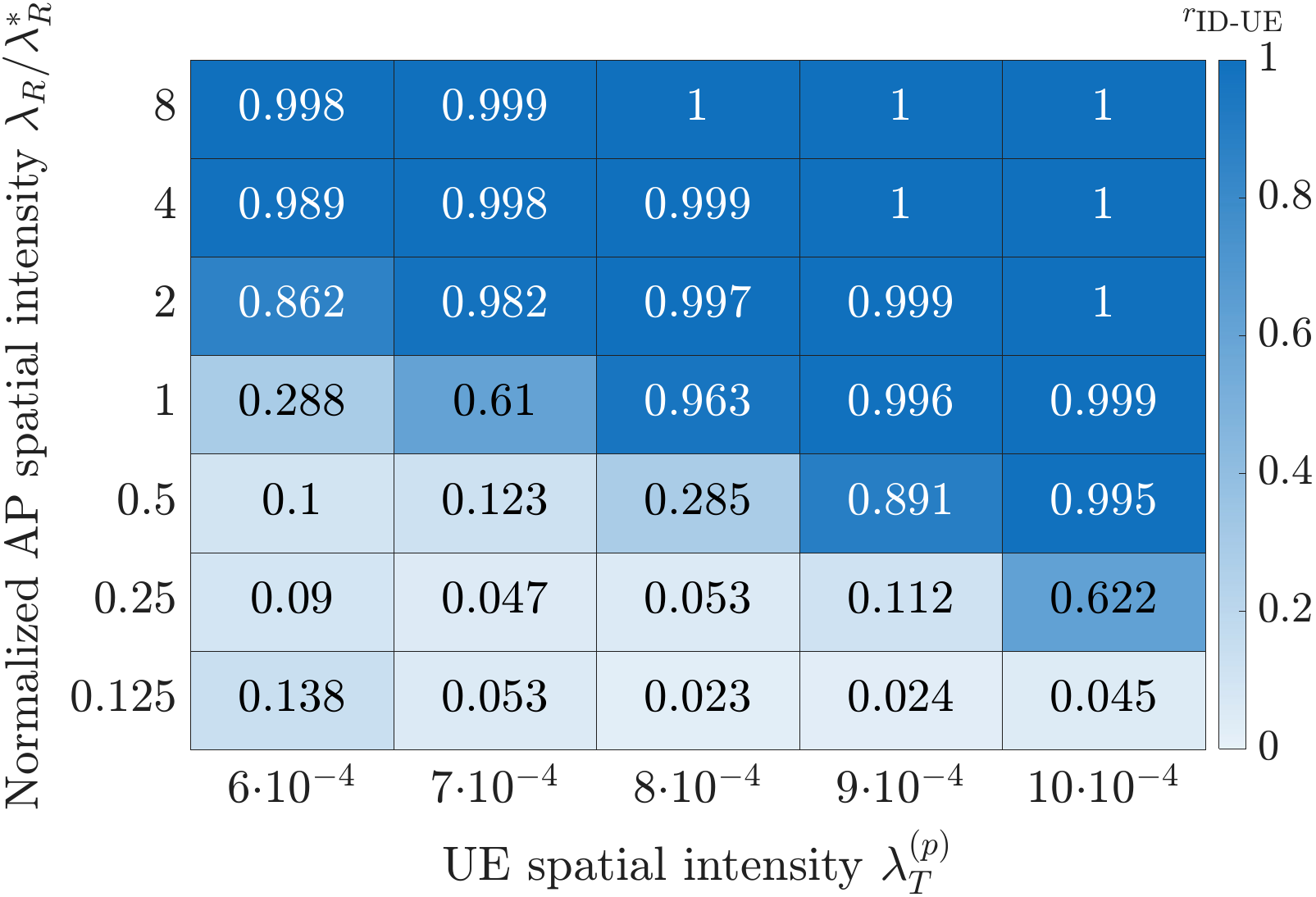}
        \label{fig:r_id_ue_geometric_gamma50}}
    \caption{Identifiability rates for $\gamma=50$.}
    \label{fig:r_id_gamma50}
    \vspace*{-6mm}
\end{figure*}
\begin{figure*}[t]
    \centering
    \subfloat[$r_\text{ID}$ for independent-edge graphs.]{\includegraphics[width=\plotwidthFull\textwidth]{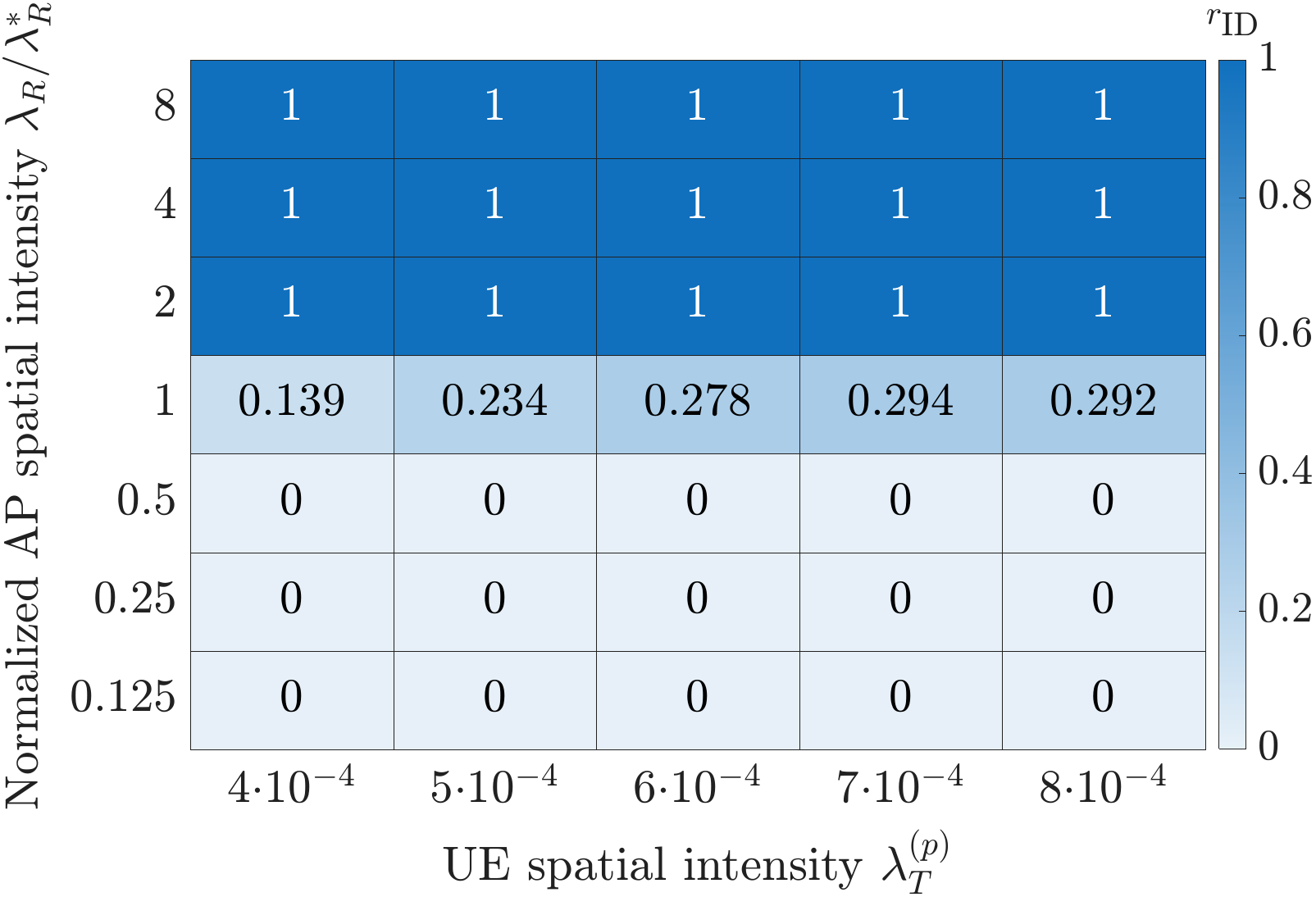}
        \label{fig:r_id_independent_gamma70}}\hfill
    \subfloat[$r_\text{ID}$ for geometric graphs.]{\includegraphics[width=\plotwidthFull\textwidth]{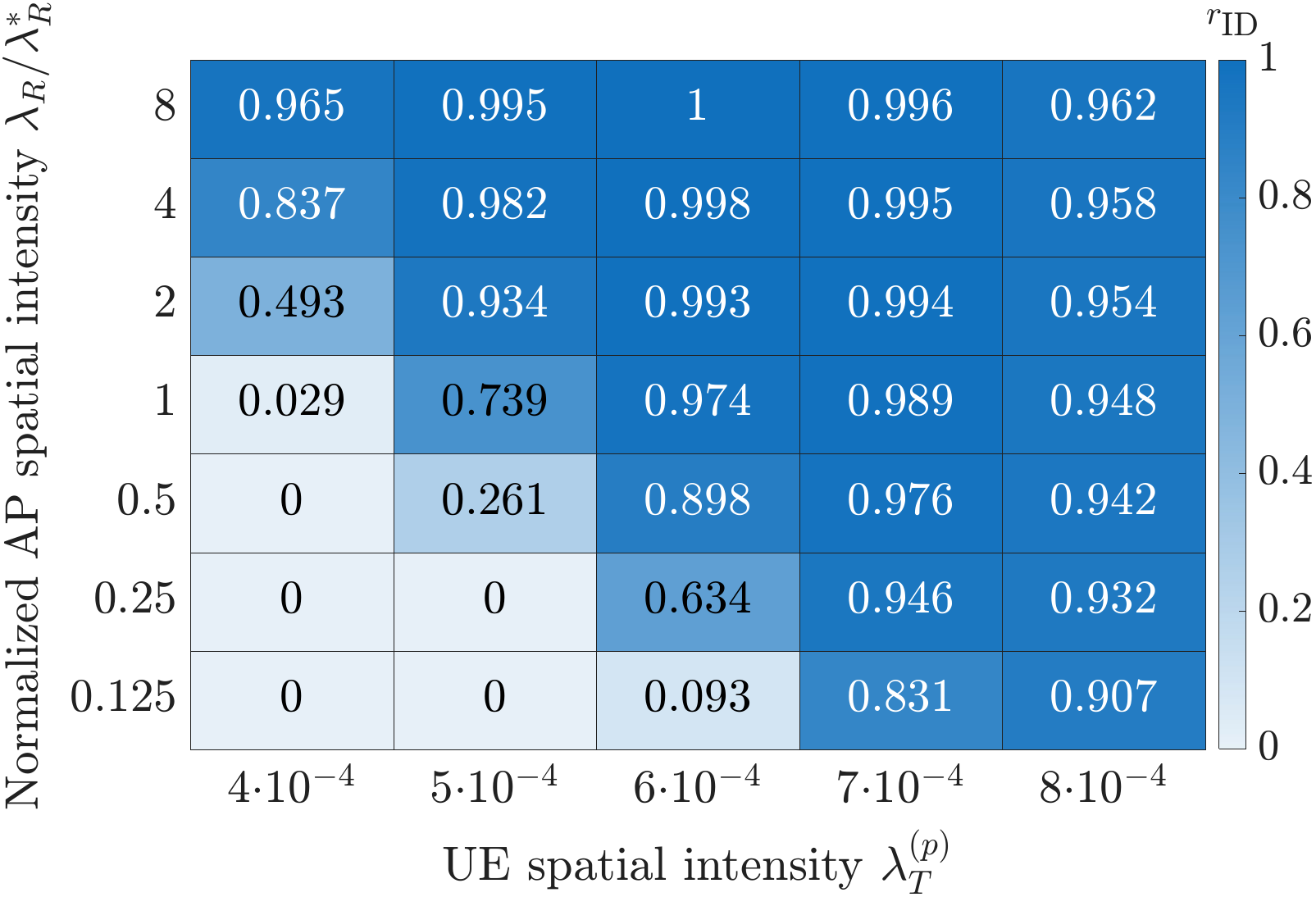}
        \label{fig:r_id_geometric_gamma70}}\hfill
    \subfloat[$r_\text{ID-UE}$ for geometric graphs.]{\includegraphics[width=\plotwidthFull\textwidth]{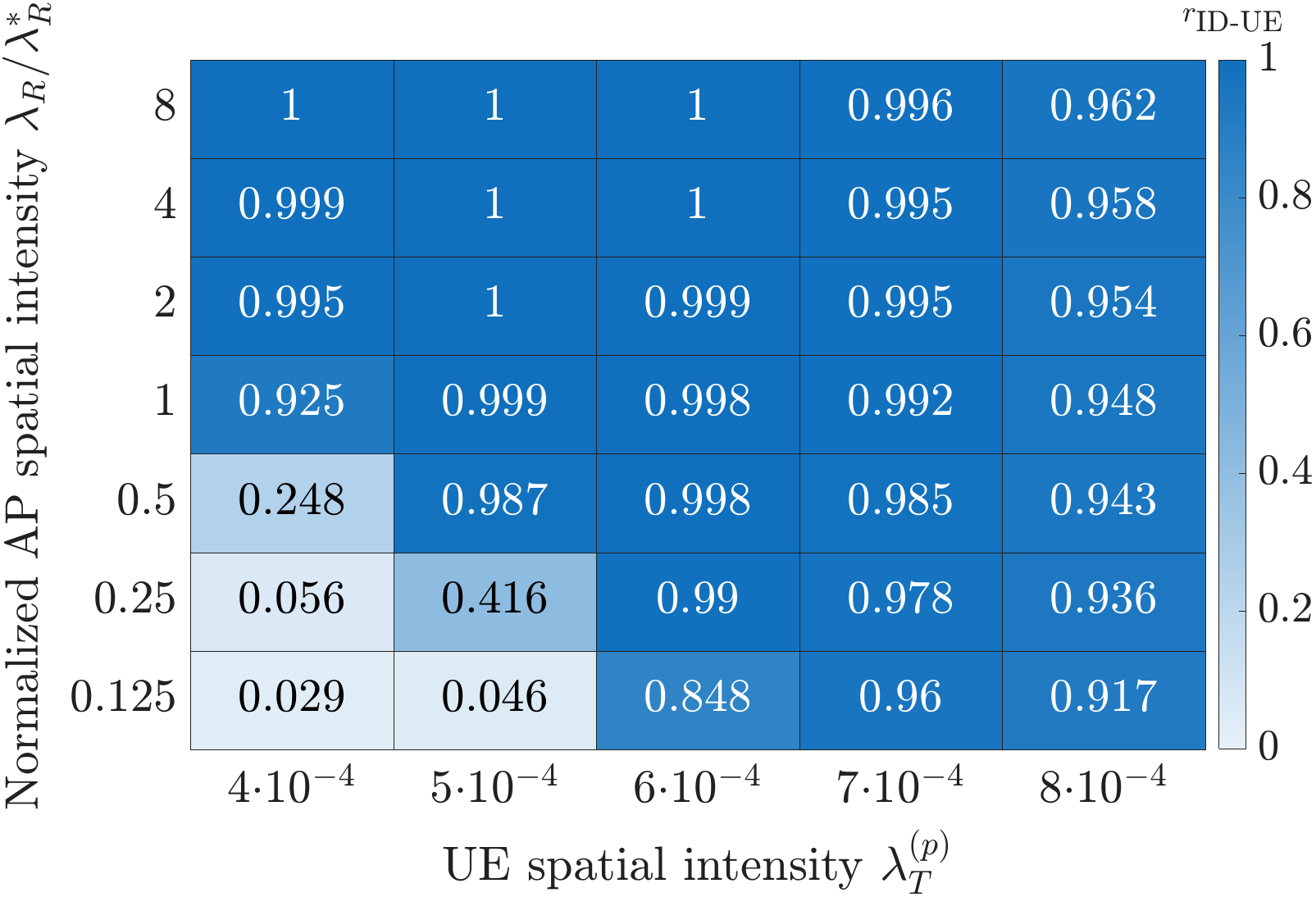}
        \label{fig:r_id_ue_geometric_gamma70}}
    \caption{Identifiability rates for $\gamma=70$.}
    \label{fig:r_id_gamma70}
    \vspace*{-4mm}
\end{figure*}
For graphs with independently generated edges, the identifiability rate $r_\text{ID}$ is equal to one for $\lambda_R>\lambda_R^*$ and zero for $\lambda_R<\lambda_R^*$, showing the sharp phase-transition behavior predicted in Sections~\ref{sec:ID_PPP} and~\ref{sec:identifiability_region}.
In contrast, for \acp{BRGG}, the phase transition is less sharp and the threshold value $\lambda_R^*$ differs from the previous class of graphs, depending on the \ac{UE} spatial density $\lambda_T^{(p)}$ and $\gamma$.
For small $\Lambda_T^{(p)}$, the \ac{AP} spatial density required to ensure a high identifiability rate is larger than the one predicted for independent-edge graphs, determined by~\eqref{eq:Lambda_R_crit}.
In contrast, for large $\Lambda_T^{(p)}$, the critical \ac{AP} spatial density is smaller for \acp{BRGG} than for graphs with independent edges.
Finally, when considering the per-UE identifiability rate $r_\text{ID-UE}$, we note that the required \ac{AP} density can be lowered in practice since usually it suffices to ensure good performance for the majority of the \acp{UE}, rather than for all of them.
For independent-edge graphs, the identifiability rate per \ac{UE} $r_\text{ID-UE}$ is not shown but, similarly to the identifiability rate $r_\text{ID}$, it remains low for $\lambda_R<\lambda_R^*$.

\vspace*{-1mm}
\section{Conclusion}\label{sec:concl}
In this work, we analyzed the identifiability of semi-blind estimation of user channels and data signals in \ac{CF-MaMIMO} networks.
To this end, we focused on \ac{CF-MaMIMO} networks in which \acp{AP} and \acp{UE} are spatially deployed according to \acp{PPP}.
To enable a tractable asymptotic identifiability analysis, the resulting \acp{BRGG} modeling  such networks were approximated  by associated independent-edge random graphs. 
Within this framework, we characterized the identifiability region  as a function of macroscopic network parameters, namely the \ac{UE} and \ac{AP} spatial intensities and the neighborhood radius which defines the distance beyond which the channels are assumed to be negligible.
Finally, we provided Monte Carlo simulations  to verify the derived identifiability region for graphs with independent edges and assess the extent to which the surrogate model accurately captures the properties of \acp{BRGG}.

\vspace*{-1mm}
\linespread{0.87}
\bibliographystyle{IEEEtran}
\bibliography{IEEEabrv,references_new}

\end{document}

%% file: abbreviations.tex
\begin{acronym}
\acro{5G}{fifth generation}
\acro{6G}{sixth generation}
\acro{AMP}{approximate message passing}
\acro{AP}{access point}
\acro{APP}{a posteriori probability}
\acro{AWGN}{additive white Gaussian noise}
\acro{B5G}{beyond fifth generation}
\acro{BEC}{binary erasure channel}
\acro{BG}{Bernoulli-Gaussian}
\acro{BPSK}{binary phase-shift-keying}
\acro{BRGG}{bipartite random geometric graph}
\acro{CDF}{cumulative distribution function}
\acro{CF-MaMIMO}{cell-free massive multiple-input multiple-output}
\acro{CPU}{central processing unit}
\acro{CSI}{channel state information}
\acro{DER}{detection error rate}
\acro{DFT}{discrete Fourier transform}
\acro{EP}{expectation propagation}
\acro{GaBP}{Gaussian belief propagation}
\acro{GF-CF-MaMIMO}{grant-free cell-free massive multiple-input multiple-output}
\acro{iid}[i.i.d.]{independent and identically distributed}
\acro{IoT}{Internet-of-Things}
\acro{JACD}{joint activity detection, channel estimation, and data detection}
\acro{JACD-EP}{joint activity detection, channel estimation, and data detection via expectation propagation}
\acro{JAC}{joint activity detection and channel estimation}
\acro{JAC-EP}{joint activity detection and channel estimation via expectation propagation}
\acro{JACD-EP-BG}{joint activity detection, channel estimation, and data detection via expectation propagation with Bernoulli-Gaussian distributions}
\acro{JCD}{joint channel estimation and data detection}
\acro{KL}{Kullback-Leibler}
\acro{LDPC}{low-density parity-check}
\acro{LSFC}{large-scale fading coefficient}
\acro{MaMIMO}{massive multiple-input multiple-output}
\acro{MIMO}{multiple-input multiple-output}
\acro{MAP}{maximum a posteriori}
\acro{MMSE}{minimum mean squared error}
\acro{MMV-AMP}{multiple measurement vector approximate message passing}
\acro{NMSE}{normalized mean squared error}
\acro{PC}{pilot contamination}
\acro{PPP}{Poisson point process}
\acro{QAM}{quadrature amplitude modulation}
\acro{QoS}{quality of service}
\acro{RGG}{random geometric graph}
\acro{SER}{symbol error rate}
\acro{UE}{user equipment}
\acro{UPP}{uniform point process}
\acro{VB}{variational Bayes}
\acro{VB-BP-EP}[VB-EP]{variational Bayes, belief propagation, and expectation propagation}
\acro{VPS}{virtual pilot sequence}
\acro{VPM}{virtual pilot matrix}
\end{acronym}

%% file: macros.tex
\newcommand{\lvec}[1]{\ensuremath{\mathbf{#1}}}  		
\newcommand{\lmat}[1]{\ensuremath{\mathbf{#1}}}  		
\newcommand{\er}{\mbox{\ensuremath{\mathrm{e}}}}
\newcommand{\Cset}{\mathbb{C}}
\newcommand{\Rset}{\mathbb{R}}

\newcommand{\pilotsymb}{\ensuremath{p}}
\newcommand{\datasymb}{\ensuremath{d}}
\newcommand{\pilot}[1]{\ensuremath{#1^\pilotsymb}}
\newcommand{\data}[1]{\ensuremath{#1^\datasymb}}